\documentclass[10pt]{llncs} 
\usepackage[latin1]{inputenc}
\usepackage{latexsym}
\usepackage{amssymb}
\usepackage{amsmath}
\usepackage{xspace}
\usepackage{gastex}

\newtheorem{fact}{Fact}
\newcommand{\eat}[1]{}

\renewcommand{\epsilon}{\varepsilon}
\renewcommand{\phi}{\varphi}

\newcommand{\defin}[1]{\textbf{#1}}
\newcommand{\tuple}[1]{\langle{#1}\rangle}
\newcommand{\distribution}[1]{\mathcal{D}(#1)}

\newcommand{\arena}{\mathcal{A}}

\newcommand{\game}{\mathbb{G}}
\newcommand{\Eve}{Eve\xspace}
\newcommand{\Adam}{Adam\xspace}
\newcommand{\Evei}{E}
\newcommand{\Adami}{A}
\newcommand{\states}{S}
\newcommand{\actions}{\Sigma}
\newcommand{\action}{\sigma}
\newcommand{\ftrans}{\delta}

\newcommand{\requiv}{\sim}

\newcommand{\plays}[1]{Plays(#1)}

\newcommand{\outcomes}[3]{Outcomes(#1,#2,#3)}
\newcommand{\outcomesunetdemi}[2]{Outcomes(#1,#2)}
\newcommand{\mesure}[3]{\mu_{#1}^{#2,#3}}
\newcommand{\cone}[1]{cone(#1)}
\newcommand{\cones}{Cones}
\newcommand{\field}{\mathcal{F}}
\newcommand{\proba}[3]{\mathrm{Pr}_{#1}^{#2,#3}}

\newcommand{\objective}{\mathcal{O}}

\newcommand{\fstates}{F}
\newcommand{\updateK}[1]{\mathrm{UpKnow}_{#1}}
\newcommand{\knowledge}[1]{\mathrm{Knowledge}_{#1}}
\newcommand{\Ki}{K}
\newcommand{\wknowledge}[1]{\mathcal{K}^{\mathrm{AS}}_{#1}}
\newcommand{\allow}[1]{\mathrm{Allow(#1)}}
\newcommand{\AS}{\mathrm{AS}}
\newcommand{\rank}[1]{\mathrm{Rank_{#1}}}
\newcommand{\post}[3]{\mathrm{Post_{#1,#2}(#3)}}
\newcommand{\ranketoile}{\mathrm{Rank}^*}

\newcommand{\wstates}{\states^{\mathrm{AS}}}
\newcommand{\poswstates}{\states^{\mathrm{>0}}}
\newcommand{\finpreuve}[1]{\null\hfill$\qed_{[#1]}$}
\newcommand{\onehalf}{$1\frac{1}{2}$}

\title{Qualitative Concurrent Stochastic Games with Imperfect Information\thanks{Supported by the \textsc{anr} project \textsc{jade} and by the \textsc{esf} project \textsc{gasics}.}}
\author{Vincent Gripon \and Olivier Serre}
\institute{LIAFA (CNRS \& Universit\'e Paris Diderot -- Paris 7)}

\begin{document}
\maketitle

\begin{abstract}
We study a model of games that combines concurrency, imperfect information and stochastic aspects. Those are finite states games in which, at each round, the two players choose, \emph{simultaneously} and \emph{independently}, an action. Then a successor state is chosen accordingly to some fixed probability distribution depending on the previous state and on the pair of actions chosen by the players. Imperfect information is modeled as follows: both players have an equivalence relation over states and, instead of observing the exact state, they only know to which equivalence class it belongs. Therefore, if two partial plays are indistinguishable by some player, he should behave the same in both of them. We consider reachability (does the play eventually visit a final state?) and B\"uchi objective (does the play visit infinitely often a final state?). 

Our main contribution is to prove that the following problem is complete for $2$-\textsc{ExpTime}: decide whether the first player has a strategy that ensures her to almost-surely win against \emph{any} possible strategy of her oponent. We also characterise those strategies needed by the first player to almost-surely win.
\end{abstract}

\section{Introduction}

Perfect information turn based two-player games on a graph \cite{dag01} are widely studied in computer science. Indeed, they are a useful tool for both theoretical (for instance the modern proofs of Rabin's complementation lemma rely on the memoryless determinacy of parity games \cite{GH82}) and more practical applications. On the practical side, a major application of games is for the verification of reactive open systems. Those are systems composed of both a program and some (possibly hostile) environment. The verification problem consists of deciding whether the program can be restricted so that the system meets some given specification whatever the environment does. Here, restricting the system means synthesizing some controller, which, in term of games, is equivalent to designing a winning strategy for the player modeling the program \cite{RW87}.

The perfect information turn-based model, even if it suffices in many situations, is somewhat weak for the following two reasons. First, it does not permit to capture the behavior of real  concurrent models where, in each step, the program and its environment \emph{independently} choose moves, whose \emph{parallel} execution determines the next state of the system. Second, in this model both players have, at each time, a perfect information on the current state of the play: this, for instance, forbids to model a system where the program and the environment share some public variables while having also their own private variables \cite{Reif84}.

In this paper, we remove those two restrictions by considering concurrent stochastic games with imperfect information. Those are finite states games in which, at each round, the two players choose \emph{simultaneously} and \emph{independently} an action. Then a successor state is chosen accordingly to some fixed probability distribution depending on the previous state and on the pair of actions chosen by the players. Imperfect information is modeled as follows: both players have an equivalence relation over states and, instead of observing the exact state, they only see to which equivalence class it belongs. Therefore, if two partial plays are indistinguishable by some player, he should behave the same in both of them. Note that this model naturally captures several model studied in the literature \cite{dAH00,CH05,chatterjeePHD,CDHR07}. The winning conditions we consider here are reachability (is there a final state eventually visited?), B\"uchi (is there a final state that is visited infinitely often?) and their dual versions, safety and co-B\"uchi.

We study qualitative properties of those games (note that quantitative properties --- \emph{e.g.} deciding whether the value of the game is above a given threshold --- are already undecidable in much weaker models \cite{Paz71}). More precisely, we investigate the question of deciding whether some player can almost-surely win, that is whether he has a strategy that wins with probability $1$ against any counter strategy of the oponent. 
Our main contributions is to prove that, for both reachability and B\"uchi objectives, one can decide, in doubly exponential time (which is proved to be optimal), whether the first player has an almost-surely winning strategy. Moreover, when it is the case, we are also able to construct such a finite-memory strategy. 
We also provide intermediate new results concerning positive winning in safety (and co-B\"uchi) \onehalf-player games (\emph{a.k.a} partial observation Markov decision process).
\smallskip

\noindent {\bf Related work.} Concurrent games with perfect information have been deeply investigated in the last decade \cite{HdAK97,dAH00,chatterjeePHD}. Games with imperfect information have been considered for turn-based model \cite{Reif84} as well as for concurrent models with only one imperfectly informed player \cite{CH05,CDHR07}. 
To our knowledge, the present paper provides the first positive results on a model of games that combines concurrency, imperfect information (on both sides) and stochastic transition function. 
In a recent independent work \cite{PrivateGimbert}, Bertrand, Genest and Gimbert obtain similar results than the one presented here for a closely related model. 
Bertand \emph{et al.} also discuss qualitative determinacy results and consider the case where a player is more informed than the other. We refer the reader to \cite{PrivateGimbert} for a detailed exposition.

\noindent{\bf Comparison with our previous work \cite{GS09}.} Krishnendu Chatterjee and Laurent Doyen pointed out an important mistake in a previous version of this work originally published in ICALP'09 \cite{GS09}. Indeed, in this previous work we additionally assumed that the players did not observe the actions they played and we proposed a variation of the constructions for this richer setting. It turned out that the constructions were wrong. We refer to \cite{CD11} for more insight on this point as well as for an answer to this more general problem.

\section{Definitions}

A \defin{probability distribution} over a finite set $X$ is a mapping $d:X\rightarrow [0,1]$ such that $\displaystyle\sum_{x\in X}d(x)=1$. In the sequel we denote by $\distribution{X}$ the set of probability distributions over $X$.

Given some set $X$ and some equivalence relation $\sim$ over $X$, $[x]_{\sim}$ stands for the equivalence class of $x$ for $\sim$ and $X/_{\sim}=\{[x]_{\sim}\mid x\in X\}$ denotes the set of equivalence classes of $\sim$.

For some finite alphabet $A$, $A^*$ (\emph{resp.} $A^\omega$) designates the set of finite (\emph{resp.} infinite) words over $A$.

\subsection{Arenas}

A \defin{concurrent arena with imperfect information} is a tuple $\arena = \langle\states,\actions_\Evei,\actions_\Adami,$\\$\ftrans,\requiv_\Evei,\requiv_\Adami\rangle$
where 
\begin{itemize}
	\item $\states$ is a finite set of control states; 
	\item $\actions_\Evei$ (\emph{resp.} $\actions_\Adami$) is the (finite) set of actions for \Eve (\emph{resp.} \Adam);
	\item $\ftrans: \states\times\actions_\Evei\times \actions_\Adami \rightarrow \distribution{\states}$ is the transition (total) function;
	\item $\requiv_\Evei$ and $\requiv_\Adami$ are two equivalence relations over states.
\end{itemize}
A play in a such an arena proceeds as follows. First it starts in some initial state $s$. Then \Eve picks an action $\action_\Evei\in\actions_\Evei$ and, \emph{simultaneously} and \emph{independently}, \Adam chooses an action $\action_\Adami\in\actions_\Adami$. Then a successor state is chosen accordingly to the probability distribution $\delta(s,\action_\Evei,\action_\Adami)$. Then the process restarts: the players choose a new pair of actions that induces, together with the current state, a new state and so on forever. Hence a \defin{play} is an infinite sequence $s_0(\action_{\Evei,0},\action_{\Adami,0})s_1(\action_{\Evei,1},\action_{\Adami,1})s_2(\action_{\Evei,2},\action_{\Adami,2})\cdots$ in $(\states\cdot(\actions_\Evei\times \actions_\Adami))^\omega$ such that for every $i\geq 0$, $\delta(s_i,\action_{\Evei,i},\action_{\Adami,i})(s_{i+1})>0$. In the sequel we refer to a prefix of a play (ending in a state) as a \defin{partial play} and we denote by $\plays{\arena}$ the set of all plays in arena $\arena$.

The intuitive meaning of $\requiv_\Evei$ (\emph{resp.} $\requiv_\Adami$) is that two states $s_1$ and $s_2$ such that $s_1\requiv_\Evei s_2$ (\emph{resp.} $s_1\requiv_\Adami s_2$) cannot be distinguished by \Eve (\emph{resp.} by \Adam). We easily extend the relation $\requiv_\Evei$ to partial plays (here Eve observes her actions, but does not observe Adam's actions): let $\lambda=s_0(\action_{\Evei,0},\action_{\Adami,0})s_1(\action_{\Evei,1},\action_{\Adami,1})\cdots s_n$ and $\lambda'=s_0'(\action_{\Evei,0}',\action_{\Adami,0}')s_1'(\action_{\Evei,1}',\action_{\Adami,1}')\cdots s_n'$ be two partial plays, then $\lambda\requiv_\Evei \lambda'$ if and only if $s_i\requiv_\Evei s'_i$ and $\action_{\Evei,i}=\action_{\Evei,i}'$ for all $i=0,\cdots,n$. 

Note that perfect information concurrent arenas (in the sense of \cite{HdAK97,dAH00}) correspond to the special case where $\requiv_\Evei$ and $\requiv_\Adami$ are the equality relation over $\states$.

\subsection{Strategies}

In order to choose their moves the players follow strategies, and, for this, they may use all the information they have about what was played so far. However, if two partial plays are equivalent for $\requiv_\Evei$, then \Eve cannot distinguish them, and should therefore behave the same. This leads to the following notion.

An \defin{observation-based strategy} for \Eve is a function $\phi_\Evei: (\states/_{\requiv_\Evei})\cdot(\actions_\Evei\cdot\states/_{\requiv_\Evei} )^*\rightarrow \distribution{\actions_\Evei}$,  \emph{i.e.}, to choose her next action, Eve considers the sequence of observations she got so far. In particular, a strategy $\phi_\Evei$ is such that $\phi_\Evei(\lambda)=\phi_\Evei(\lambda')$ whenever $\lambda\requiv_\Evei \lambda'$. Observation-based strategies for \Adam are defined similarly. 

Of special interest are those strategies that does not require memory: a \defin{memoryless observation-based strategies} for \Eve is a function from $\states/_{\requiv_\Evei}\rightarrow \distribution{\actions_\Evei}$, that is to say these strategies only depend of the current equivalence class.

A \defin{uniform strategy} for some player $X$ is a strategy $\phi$ such that for all partial play $\lambda$, the probability measure $\phi(\lambda)$ is uniform, \emph{i.e.}, for all action $\action_X\in\actions_X$, either $\phi(\lambda)(\action_X)=0$ or $\phi(\lambda)(\action_X)=\frac{1}{|\{\action_X\in\actions_X\mid \phi(\lambda)(\action_X)\neq 0\}|}$.
The set of memoryless uniform strategies for $X$ is a finite set containing $(2^{|\actions_X|}-1)^{|\states|}$ elements. Equivalently those strategies can be seen as functions to (non-empty) sets of (authorised) actions.

A \defin{finite-memory strategy} for \Eve with memory $M$ ($M$ being a finite set) is some triple $\phi=(Move,Up,m_0)$ where $m_0\in M$ is the initial memory, $Move:M\rightarrow \distribution{\actions_\Evei}$ associates a distribution of actions with any element in the memory $M$ and $Up:M\times \states/_{\requiv_\Evei} \times \actions_\Evei\rightarrow M$ is a mapping updating the memory with respect to some observation (and the last action played by \Eve). One defines $\phi(s_0)=Move(m_0)$ and $\phi(s_0(\action_{\Evei,0},\action_{\Adami,0})\cdots s_n)=Move(Up(\cdots Up(Up(m_0,[s_1]/_{\requiv_\Evei},\action_{\Evei,0}),[s_2]/_{\requiv_\Evei},\action_{\Evei,1}),\cdots,[s_n]/_{\requiv_\Evei},\action_{\Evei,n-1})\cdots)$ for any $n\geq 1$. Hence, a finite-memory strategy is some observation-based strategy that can be implemented by a finite transducer whose set of control states is $M$.

%

\subsection{Probability Space and Outcomes of Strategies}\label{subsection:proba}

Let $\arena = \tuple{\states,\actions_\Evei,\actions_\Adami,\ftrans,\requiv_\Evei,\requiv_\Adami}$ be a concurrent arena with imperfect information, let $s_0\in \states$ be an initial state, $\phi_\Evei$ be a strategy for \Eve and $\phi_\Adami$ be a strategy for \Adam. In the sequel we are interested in defining the probability of a (measurable) set of plays knowing that \Eve (\emph{resp.} \Adam) plays accordingly $\phi_\Evei$ (\emph{resp.} $\phi_\Adami$). This is done in the classical way: first one defines the probability measure for basic sets of plays (called here \emph{cones} and corresponding to plays having some initial common prefix) and then extends it in a unique way to all measurable sets.

First define $\outcomes{s_0}{\phi_\Evei}{\phi_\Adami}$ to be the set of all possible plays when the game starts on $s_0$ and when \Eve and \Adam plays respectively accordingly to $\phi_\Evei$ and $\phi_\Adami$. More formally, an infinite play $\lambda =s_0(\action_{\Evei,0},\action_{\Adami,0})s_1(\action_{\Evei,1},\action_{\Adami,1})s_2\cdots$  belongs to $\outcomes{s_0}{\phi_\Evei}{\phi_\Adami}$ if and only if, for every $i\geq 0$, 
$\phi_\Evei(s_0\action_{\Evei,0}s_1\action_{\Evei,1}\cdots s_i)(\action_{\Evei,i})>0$ and $\phi_\Adami(s_0\action_{\Adami,0}s_1\action_{\Adami,1}\cdots s_i)(\action_{\Adami,i})>0$ (\emph{i.e.} $\action_{X}$ is possible accordingly to $\phi_X$, for $X=\Evei,\Adami$).

Now, for any partial play $\lambda$, the \defin{cone} for $\lambda$ is the set $\cone{\lambda}=\lambda\cdot ((\actions_\Evei\times\actions_\Adami)\cdot\states)^\omega$ of all infinite plays with prefix $\lambda$. Denote by $\cones$ the set of all possible cones and let $\field$ be the Borel $\sigma$-field generated by $\cones$ considered as a set of basic open sets (\emph{i.e.} $\field$ is the smallest set containing $\cones$ and closed under complementation, countable union and countable intersection). Then $(\plays{\arena},\field)$ is a $\sigma$-algebra.

A pair of strategies $(\phi_\Evei,\phi_\Adami)$ induces a probability space over $(\plays{\arena},\field)$. Indeed one can define a measure $\mesure{s_0}{\phi_\Evei}{\phi_\Adami}:\cones \rightarrow [0,1]$ on cones (this task is easy as a cone is uniquely defined by a finite partial play) and then uniquely extend it to a probability measure on $\field$ using the Carath\'eodory Unique Extension Theorem. For this, one defines $\mesure{s_0}{\phi_\Evei}{\phi_\Adami}$ inductively on cones:
\begin{itemize}
\item $\mesure{s_0}{\phi_\Evei}{\phi_\Adami}(\cone{s}) = 1$ if $s=s_0$ and $\mesure{s_0}{\phi_\Evei}{\phi_\Adami}(s) = 0$ otherwise.
\item For every partial play $\lambda$ ending in some vertex $s$, 
$$\mesure{s_0}{\phi_\Evei}{\phi_\Adami}(\cone{\lambda\cdot (\action_\Evei,\action_\Adami)\cdot s'})= 
\mesure{s_0}{\phi_\Evei}{\phi_\Adami}(\cone{\lambda}).
\phi_\Evei(\lambda)(\action_\Evei). \phi_\Adami(\lambda)(\action_\Adami).\delta(s,\sigma_\Evei,\sigma_\Adami)(s')$$
\end{itemize}

Denote by $\proba{s_0}{\phi_\Evei}{\phi_\Adami}$ the unique extension of $\mesure{s_0}{\phi_\Evei}{\phi_\Adami}$ to a probability measure on $\field$. Then $(\plays{\arena},\field,\proba{s_0}{\phi_\Evei}{\phi_\Adami})$ is a probability space. 

\subsection{Objectives, Value of a Game}\label{subsection:objectives}

Fix a concurrent arena with imperfect information $\arena$. An objective for \Eve is a measurable set $\objective\subseteq \plays{\arena}$: a play is won by her if it belongs to $\objective$; otherwise it is won by \Adam. A \defin{concurrent game with imperfect information} is a triple $(\arena,s_0,\objective)$ where $\arena$ is a concurrent arena with imperfect information, $s_0$ is an initial state and $\objective$ is an objective. 
In the sequel we focus on the following special classes of objectives (note that all of them are Borel sets hence measurable) that we define as means of a subset $F\subseteq \states$ of \defin{final states}.

\begin{itemize}
\item\defin{Reachability objective}: a play is winning if it eventually goes through some final state.
\item \defin{Safety objective}: a play is winning if it never goes through a final state.
\item \defin{B\"uchi objective}: a play is winning if it goes infinitely often through final states.
\item \defin{Co-B\"uchi objective}: a play is winning if it goes finitely often through final states.
\end{itemize}

A reachability (\emph{resp.} safety, B\"uchi, co-B\"uchi) game is a game equipped with a reachability (\emph{resp.} safety, B\"uchi, co-B\"uchi) objective. In the sequel we may replace $\mathcal{O}$ by $F$ when it is clear from the context which winning condition we consider.

Fix a concurrent game with imperfect information $\game= (\arena,s_0,\objective)$. A strategy $\phi_\Evei$ for \Eve is \defin{almost-surely winning} if, for any counter-strategy $\phi_\Adami$ for \Adam, $\proba{s_0}{\phi_\Evei}{\phi_\Adami}(\objective)=1$. If such a strategy exists, we say that \Eve \defin{almost-surely wins} $\game$.
A strategy $\phi_\Evei$ for \Eve is \defin{positively winning} if, for any counter-strategy $\phi_\Adami$ for \Adam, $\proba{s_0}{\phi_\Evei}{\phi_\Adami}(\objective)>0$. If such a strategy exists, we say that \Eve \defin{positively wins} $\game$.

\section{Knowledge Arena}

For the rest of this section we let $\arena$ be a concurrent arena with imperfect information with $\arena = \tuple{\states,\actions_\Evei,\actions_\Adami,\ftrans,\requiv_\Evei,\requiv_\Adami}$ and let $s_0\in \states$ be some initial state.


Even, if she does not see the precise control state, \Eve can deduce information about it from previous information on the control state and from the action she just played. We should refer to this as the \defin{knowledge} of \Eve, which formally is a set of states.
Assume \Eve knows that the current state belongs to some set $K\subseteq \states$. After the next move \Eve observes the equivalence class $[s]_{\requiv_{\Evei}}$ of the new control state and she also knows that she played action $\action_\Evei$: hence she can compute the set of possible states the play can be in. This is done using the function $\updateK{}:2^\states\times [\states]_{/_{\requiv_{\Evei}}}\times {\actions_\Evei}\rightarrow 2^{\actions_\states}$ defined by letting
$$
   \updateK{}(K,[s]_{\requiv_{\Evei}},\action_\Evei) = \{t\requiv_{\Evei} s \mid \exists r\in K,\ \exists   \action_\Adami\in \actions_\Adami\text{ s.t. }\ftrans(r,\action_\Evei,\action_\Adami)(t)>0\}
$$
\emph{i.e.} in order to update her current knowledge, observing in which equivalence class is the new control state, and knowing that she played $\action_\Evei$, \Eve computes the set of all states in this class that may be reached from a state in her former knowledge.

Finally, we let \begin{multline*}\knowledge{}(s_0\action_{\Evei,0}s_1\cdots s_{n})=\\
\updateK{}(\updateK{}(\cdots\updateK{}(\{s_0\},[s_1]_{\requiv_{\Evei}},\action_{\Evei,1})\cdots)
,[s_n]_{\requiv_{\Evei}},\action_{\Evei,n})\end{multline*}

Based on our initial remark and on the notion of knowledge we define the \defin{knowldege arena associated with $\arena$}, denoted $\arena^\Ki$. The arena $\arena^\Ki$ is designed to make explicit the information \Eve can collect about the possible current state (\emph{i.e.} the knowledge). 
We define $\arena^\Ki=\tuple{\states^\Ki,\actions_\Evei,\actions_\Adami,\ftrans^\Ki,\requiv_{\Evei}^{\Ki},\requiv_{\Adami}^{\Ki}}$ as follows:
\begin{itemize}
\item $\states^\Ki = \{(s,K)\in \states\times 2^\states\mid K\subseteq [s]_{/\requiv{\Evei}}\}$: the first component is the real state and the second one is the current knowledge of \Eve;
\item $\ftrans^\Ki((s,K),\action_\Evei,\action_\Adami)(s',K')=0$ if $K'\neq \updateK{}(K,$ $[s']_{\requiv_{\Evei}},\action_\Evei)$;\\ and $\ftrans^\Ki((s,K),\action_\Evei,\action_\Adami)(s',K')=\ftrans(s,\action_\Evei,\action_\Adami)(s')$ otherwise: $\ftrans^\Ki$ behaves as $\ftrans$ on the first components and \emph{deterministically} updates the knowledge;
\item $(s,K)\requiv_{\Evei}^{\Ki} (s',K'')$ if and only if $K = K'$ (implying $s\requiv_\Evei s'$): \Eve only observes her knowledge;
\item $(s,K)\requiv_{\Adami}^{\Ki} (s',K')$ if and only if $s\requiv_\Adami s'$: \Adam does not observe \Eve's knowledge.
\end{itemize}


Consider an observation-based strategy $\phi$ for \Eve in the arena $\arena$. Then it can be converted into an observation-based strategy on the associated knowledge arena. For this, remark that in the knowledge arena, those states reachable from the initial state $(s_0,\{s_0\})$ are of the form $(s,K)$ with all states in $K$ being equivalent with $s$ with respect to $\requiv_\Evei$. Then one can define 
$$\phi^\Ki((s_0,K_0)\action_{\Evei,0}(s_1,K_1)\action_{\Evei,1}\cdots (s_n,K_n))=\phi([s_0]_{\requiv_\Evei}\action_{\Evei,0}[s_1]_{\requiv_\Evei}\action_{\Evei,1}\cdots[s_n]_{\requiv_\Evei})$$

Note that $\phi^\Ki$ is observation-based as, for all $0\leq h\leq n$, $[s_h]_{\requiv_\Evei}$ is uniquely defined from the $K_h$, that are observed by \Eve in the knowledge arena. 

Conversely, any observation-based strategy in the knowledge arena can be converted into an observation-based strategy in the original arena. Indeed, consider some observation-based strategy $\phi^\Ki$ in the knowledge arena: it is a mapping from $2^\states\cdot(\actions_\Evei\cdot 2^\states)^*$ into $\distribution{\actions_\Evei}$ (the equivalent classes of the relation $\requiv_\Evei^\Ki$ are, by definition, isomorphic with $2^\states$). Now, note that \Eve can, while playing in $\arena$, compute on the fly her current knowledge (applying function $\updateK{}$ to her previous knowledge and to the last action she played): hence along a play $s_0(\action_{\Evei,0},\action_{\Adami,0})s_1(\action_{\Evei,1},\action_{\Adami,1})\cdots s_n$ she can compute the corresponding sequence $K_0K_1\cdots K_n$ of knowledges.
Now it suffices to consider the observation-based strategy $\phi$ for \Eve in the initial arena defined by:
$$\phi(s_0\action_{\Evei,0} s_1\action_{\Evei,1}\cdots s_n)  = \phi^\Ki(K_0\action_{\Evei,0}K_1\action_{\Evei,1}\cdots K_n)$$

Note that this last transformation (taking a strategy $\phi^\Ki$ and producing a strategy $\phi$) is the inverse of the first transformation (taking a strategy $\phi$ and producing a strategy $\phi^\Ki$). In particular, it proves that the observation-based strategies in both arena are in bijection.
It should be clear that those strategies for \Adam in both games are the same (as what he observes is identical). 

Assume that $\arena$ is equipped with a set $F$ of final states. Then one defines the final states in $\arena^\Ki$ by letting $F^\Ki=\{(f,K)\mid f\in F\}\cap \states^\Ki$: this allows to define an objective $\objective^\Ki$ in $\arena^\Ki$ from an objective $\objective$ in $\arena$. Based on the previous observations, we derive the following.

\begin{proposition}
Let $\game=(\arena,s_0,\objective)$ be some imperfect information game equip\-ped with a reachability (\emph{resp.} saftey, B\"uchi, co-B\"uchi) objective. Let $\game^\Ki=(\arena^\Ki,(s_0,\{s_0\}),\objective^\Ki)$ be the associated game played on the knowledge arena. Then for any strategies $\phi_\Evei,\phi_\Adami$ for \Eve and \Adam, the following holds:\\
$\proba{s_0}{\phi_\Evei}{\phi_\Adami}(\objective) = \proba{(s_0,\{s_0\})}{\phi_\Evei^\Ki}{\phi_\Adami}(\objective^\Ki)$.
In particular, \Eve has an almost-surely winning observation-based strategy in $\game$ if and only if she has one in $\game^\Ki$. 
\end{proposition}

In the setting of the previous proposition, consider the special case where \Eve has an almost-surely winning observation-based strategy $\phi^\Ki$ in $\game^\Ki$ that only depends on the current knowledge (in particular, it is \emph{memoryless}). Then the corresponding almost-surely winning observation-based strategy $\phi$ in $\game$ is, in general, not memoryless, but can be implemented by a finite transducer whose set of control states is precisely the set of possible knowledges for \Eve. More precisely the strategy consists in computing and updating on the fly (using a finite automaton) the value of the knowledge after the current partial play and to pick the next action by solely considering the knowledge. We may refer at such a strategy $\phi$ as a \defin{knowledge-only strategy}.

\section{Reachability Objectives}

The main result of this section is the following.

\begin{theorem}\label{theo:reachability}
For any reachability concurrent game with imperfect information, one can decide, in doubly exponential time, whether \Eve has an almost-surely winning strategy. If \Eve has such a strategy then she has a knowledge-only uniform strategy, and such a strategy can be effectively constructed.
\end{theorem}

Before proving Theorem~\ref{theo:reachability} we first establish an intermediate result on positively winning in \onehalf-player safety game

\subsection{Positively winning in \onehalf-player safety game with imperfect information}

A concurrent game (with imperfect information) in which one player has only a single available action is what we refer as a \defin{\onehalf-player game with imperfect information} (those games are also known in the literature as \emph{partially observable Markov Decision Processes}). The following result is a key ingredient for the proofs of Proposition~\ref{prop:uniformstratiswinning} and Theorem~\ref{theo:reachability}.

\begin{lemma}\label{lemma:safety1.5}
Consider an \onehalf-player safety game with imperfect information. Assume that the player has an observation-based strategy that is positively winning.
Then she also has an observation-based finite memory strategy that is positively winning.
Moreover, both the strategy and the set of positively winning states can be computed in time $\mathcal{O}(2^{|\states|})$.
\end{lemma}


We first simplify the notations for the special case of \onehalf-player games. We define an \defin{\onehalf-player arena} as a tuple $\arena=\langle \states,\actions,\ftrans,\fstates\rangle$ where $\states$ is a finite set of control states, $\actions$ is a finite set of actions, $\fstates\subseteq \states$ is a set of final states and $\ftrans:\states\times \actions\rightarrow \distribution{\states}$ is the transition (total) function.
A \defin{\onehalf-player game} is a tuple $\game=(\arena,\requiv,s_0,\mathcal{O})$ where $\requiv$ is an equivalence relation on states, $s_0\in S$ is an initial state and $\mathcal{O}$ is an objective (safety in the sequel). The notions of knowledge, knowledge arena and knowledge-based (memoryless) strategy are trivially adapted to this setting.

We will be interested in a special kind of strategies, that we call \defin{ultimately knowledge-based memoryless strategy}. Such a strategy consists in playing randomly trying to reach a target state $t$ and then play in a knowledge-based memoryless fashion assuming that $t$ as been effectively reached.
More formally, a strategy $\phi$ is an \defin{ultimately knowledge-based memoryless strategy} if it is of the following form (where $t\in \states$ and $k<|\states|$): play uniformly randomly any action in $\actions$ on the $k$ first moves and then set knowledge to be $\{t\}$, update it along the play and pick the moves only depending on the current knowledge.
Hence, it is of the following form:
\begin{itemize}
	\item $\phi(s_0\action_{\Evei,0}s_1\cdots s_n) = d_{\mathrm{univ}}$ if $n< k$;
	\item otherwise, $\phi(s_0\action_{\Evei,0}s_1\cdots s_n) = d_{\knowledge{}(t\action_{\Evei,k}\cdot s_{k+1}\action_{\Evei,k+1}\cdots s_n)}$ where $d_{\mathrm{univ}}$ is the uniform distribution over $\actions$ and $d_K\in\distribution{\actions}$ for any $K\subseteq \states$. 
\end{itemize}

In particular those strategies are \emph{finite memory strategy} (and the memory needed is the set $\{1,\cdots k\}\cup2^\states$).

Lemma \ref{lemma:safety1.5} is a direct consequence of the following slightly more precise lemma.

\begin{lemma}\label{lemma:safety1.5bis}
Consider a \onehalf-player safety game with imperfect information. Assume that the player has an observation-based strategy that is positively winning.
Then she also has an ultimately knowledge-based memoryless strategy that is positively winning.
Moreover, both the strategy and the set of positively winning states can be computed in time $\mathcal{O}(2^{|\states|})$.\rm
\end{lemma}

\begin{proof}

In the sequel we will be interested in computing the set $\poswstates$ of those states $s\in S$ for which the player has an observation-based strategy that  is positively winning.

Consider the knowledge arena $\arena^\Ki$ and the knowledge game $\game^\Ki$ associated with $\arena$ and $\game$. We define a notion of almost-surely winning knowledge by letting
\begin{multline*}
	\wknowledge{}=\{K\in 2^\states\mid \exists \varphi \text{ knowledge-based strategy s.t. }
	\forall s\in K,\\
	\varphi \text{ is almost-surely winning for the player in } \game^\Ki \text{ from }(s,K)\}
\end{multline*}

We claim that a knowledge $K\in\wknowledge{}$ is actually surely winning: there exists a knowledge-based strategy $\phi$ such that $\forall s\in K$, $\outcomesunetdemi{\phi}{(s,K)}\subseteq ((\states\setminus F)\actions_\Evei)^\omega$, \emph{i.e.} playing accordingly to $\phi$ the player is sure that no play will be loosing for her. Indeed, consider the following (decreasing and bounded) sequence of knowledges:
$$\begin{cases}
\mathcal{K}_0 = 2^{\states\setminus F}\\
\mathcal{K}_{i+1} = \mathcal{K}_{i} \cap Pre(\mathcal{K}_{i})\\
\end{cases}
$$
where $$Pre(\mathcal{K})=\{K\in 2^{\states\setminus F}\mid \exists \action,\ \forall s\in K,\ \ftrans^\Ki((s,K),\action)((s',K'))>0\Rightarrow K'\in\mathcal{K} \}$$ 
is the set of (non-final) knowledges from which the player is sure that in the next step the play will be in $\mathcal{K}$.

Let $\mathcal{K}^*$ be the limit of the sequence $(\mathcal{K}_i)_{i\geq 0}$. Then, we have the following fact.

\begin{fact}\label{fact:wknowledge}
The following equality holds: $\wknowledge{}=\mathcal{K}^*$.
\end{fact}
\begin{proof}
The inclusion $\wknowledge{}\supseteq \mathcal{K}^*$ is immediate and it also come with a deterministic surely winning strategy for the player that simply consists in playing an action that ensures to stay inside $\mathcal{K}^*$ (such an action exists by definition of the $Pre$ operator). For the converse inclusion, we prove that a knowledge $K\notin \mathcal{K}^*$ cannot be in $\wknowledge{}$. For such a $K$ we define its rank $rk(K)=i$ to be unique interger $i$ such that $K\in \mathcal{K}_i\setminus \mathcal{K}_{i+1}$ (by convention we let $\mathcal{K}_{-1}=2^\states$) and we prove the result by induction on $rk(K)$. For $rk(K)=0$, the result is immediate. Now, assume that it holds for some $i$ and let $K$ be some knowledge with $rk(K)=i+1$: for all $\action\in\actions$, there is some $s\in\states$ and some configuration $(s',K')$ with $\ftrans^\Ki((s,K),\action)((s',K'))>0$ and $rk(K')\leq i$: hence for any knowledge-based strategy $\phi$ there is some $s\in K$ such that the probability, after one move starting from $(s,K)$, of reaching a configuration with a knowledge having a rank smaller or equal than $i$, is strictly positive. One concludes then by induction that $K\notin \wknowledge{}$.
\finpreuve{Fact~\ref{fact:wknowledge}}
\end{proof}

The following short example gives some intuition for the next construction.

\begin{example}\label{example:example1-halfgame}
Consider the arena depicted in Figure~\ref{figure:example1-halfgame} (edges are labeled by both the pair of actions and the probability of reaching their target; final states are double circled). We assume that in this example $s\requiv t\requiv t'$.
\begin{figure}[ht]
\begin{center}
\begin{picture}(40,50)(0,-10)
\gasset{AHangle=30,AHLength=3,AHlength=1.5} 
\node(A1)(20,30){$s$}
\node(A2)(0,20){$t$} 
\node(A3)(40,20){$t'$} 
\node[Nmarks=r](A4)(20,10){$f$}
\drawedge[ELside=r,ELpos=70](A1,A2){\scriptsize$(a,1/4),(b,1/4)$}
\drawedge[ELpos=70](A1,A3){\scriptsize$(a,1/4),(b,1/4)$}
\drawloop[loopangle=90](A1){\scriptsize$(a,1/4),(b,1/4)$}
\drawedge(A1,A4){\scriptsize$(b,1/4)$}
\drawedge[ELside=r](A1,A4){\scriptsize$(a,1/4)$}
\drawloop[loopangle=180](A2){\scriptsize$(a,1)$}
\drawloop[loopangle=0](A3){\scriptsize$(b,1)$}
\drawedge[ELside=r](A2,A4){\scriptsize$(b,1)$}
\drawedge(A3,A4){\scriptsize$(a,1)$}
\drawloop[loopangle=270](A4){\scriptsize$(a,1),(b,1)$}
\end{picture}\caption{Arena of Example~\ref{example:example1-halfgame}}\label{figure:example1-halfgame}
\end{center}
\end{figure}
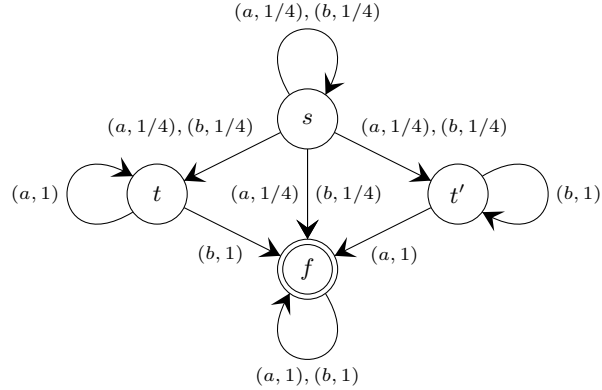

In this example $\wknowledge{}=\{\{t\},\{t'\}\}$. Note also that if one considers the corresponding knowledge game, the configurations $(t,\{t\})$ and $(t',\{t'\})$ are never visited in a play starting from $(s,\{s\})$. Nevertheless the player has a strategy that wins with probability $1/4$ starting from state $s$: the first action is to play randomly (with same probability) action $a$ or $b$, and then to play $a$ forever. The intuitive idea behind this strategy is rather simple: there are two safe states, $t$ and $t'$, that is states from which the player surely wins. Moreover those states can be possibly reached from $s$ (they belong to the same connected component) but this cannot be detected due to the equivalence relation $\requiv$. Hence the player bets that one of those states is reached (here state $t$) and behaves like if it is the case. Playing randomly on her first move is a way to ensure that with some positive probability (here $1/4$) state $t$ is reached. As from $t$ the player can surely wins, following the associated (knowledge-based memoryless) strategy, she is sure to win if her bet was correct: altogether it provides a strategy that is winning with probability $1/4$.
\finpreuve{Example~\ref{example:example1-halfgame}}
\end{example}

Define $\wstates=\{s\in\states\mid \{s\}\in\wknowledge{}\}$ and consider the following (increasing and bounded) sequence:
$$\begin{cases}
W_0 = \wstates\\
W_{i+1} = W_{i} \cup \{s\in\states\setminus F\mid \exists \action\in\actions \text{ and }t\in W_i \text{ s.t. }\delta(s,\action)(t)>0\}\\
\end{cases}
$$
Let $W$ be the limit of the sequence $(W_{i})_{i\geq 0}$: it consists exactly of those states from which the player has a strategy ensuring to reach, in at most $|\states|$ moves and without going through a final state, a state in $\wstates$ with some positive probability. Also note that the corresponding strategy is the one that plays with equal probability any action in $\actions$.

The states in $W$ are actually those from which the player has a positively winning strategy.

\begin{fact}\label{fact:poswinningstates1-2}
The following equality holds: $W=\poswstates$.
\end{fact}

\begin{proof}
The inclusion $W\subseteq \poswstates$ is rather immediate: from some state $s\in W$, the player should first play randomly on the first $k$ rounds (where $k$ is the smallest integer such that $s\in W_k$ in the previously defined sequence) and then play as if it was in some state $t\in\wstates$ (where $t$ is some reachable state from $s$ in $k$ moves accordingly to the definition of the sequence $(W_i)_{i\geq 0})$. This last step is done using an observation-based strategy that mimics the one coming with the construction of $\wknowledge{}$: the only trick here is that the player "reset" the knowledge to be $\{t\}$.

Consider now the converse inclusion: $\poswstates \subseteq W$. Let $s$ be some state in $\poswstates$ and assume, by contradiction, that $s\notin W$. In particular a play starting from $s$ will never go through a state in $\wstates$. We claim that playing accordingly to some knowledge-based strategy is almost-surely loosing for the player: indeed, as a consequence of the definition of $\wstates$ (and of $\wknowledge{}$), from a state $t\notin\wstates$, playing a knowledge-based strategy, the probability of visiting a final state in the next $2^{|\states|}$ moves is some $\epsilon>0$. Moreover such a play stays outside of $\wstates$ forever. Hence, using Borel-Cantelli Lemma, the probability that such a play never goes through a final state is $0$, meaning that it is almost-surely loosing, thus contradicting the assumption that $s\in \poswstates$. Hence, $s\in W$ which concludes the proof.
\finpreuve{Fact~\ref{fact:poswinningstates1-2}}
\end{proof}

Now one can conclude the proof of Lemma~\ref{lemma:safety1.5}: the existence of the ultimately knowledge-based memoryless strategy follows from Fact~\ref{fact:poswinningstates1-2}. Effectivity as well as complexity come from the constructive way of defining $W$ (and other intermediate objects) by means of fixpoint computations.
\finpreuve{Lemma~\ref{lemma:safety1.5bis}}
\end{proof}

\subsection{Proof of Theorem~\ref{theo:reachability}}

Fix, for the rest of this section, a concurrent game with imperfect information $\game= (\arena,s_0,\objective)$ equipped with a reachability objective $\objective$ defined from a set $\fstates$ of final states. We set $\arena = \tuple{\states,\actions_\Evei,\actions_\Adami,\ftrans,\requiv_\Evei,\requiv_\Adami}$. We also consider $\game^\Ki=(\arena^\Ki,(s_0,\{s_0\}),\objective^\Ki)$ to be the corresponding knowledge game.

To prove Theorem~\ref{theo:reachability}, one first defines (in a non constructive way) a know\-ledge-only uniform strategy $\phi$ for \Eve as follows. We let
\begin{multline*}
	\wknowledge{}=\{K\in 2^\states\mid \exists \varphi_\Evei \text{ knowledge-based strategy for \Eve s.t. } \varphi_\Evei \text{ is almost-}  \\ 
	\text{surely winning for \Eve in } \game^\Ki \text{ from any }(s,K) \text{ with } s\in K\}
\end{multline*}
be the set of knowledges made only by almost-surely winning states for \Eve (note here that we require that the almost-surely winning strategy is the same \emph{for all} states with the same knowledge). 

For every knowledge $K\in\wknowledge{}$ we define 
\begin{multline*}
	\allow{K}=\{\action_\Evei\in\actions_\Evei
	\mid \forall s\in K, \forall \action_\Adami\in\actions_\Adami,\\
	 \ftrans^\Ki((s,K),\action_\Evei,\action_\Adami)((s',K))>0 \Rightarrow K'\in\wknowledge{}\}
\end{multline*}
\begin{proposition}\label{proposition:allownonempty}
For every knowledge $K\in\wknowledge{}$, $\allow{K}\neq \emptyset$.
\end{proposition}

\begin{proof}
Consider some knowledge $K\in\wknowledge{}$ and assume by contradiction that $\allow{K}= \emptyset$. As $K\in\wknowledge{}$, there exists, by definition, some knowledge-based strategy  $\varphi_\Evei$ for \Eve in $\game^\Ki$ that is almost-surely winning from any state $(s,K)$ with $s\in K$. Strategy $\varphi_\Evei$ gives the same distribution for all partial plays of the form $(s,K)$ with $s\in K$ (\emph{i.e.} consisting of a single state indistinguishable by \Eve). Moreover there exists some action $\action_\Evei\in\actions_\Evei$ such that $\phi_\Evei(s,K)(\action_\Evei)>0$ and, as $\allow{K}= \emptyset$, there is some action $\action_\Adami\in\actions_\Adami$ and some configuration $(s',K')$ such that $\ftrans^\Ki((s,K),\action_\Evei,\action_\Adami)(s',K')>0$ and $K'\notin \wknowledge{}$. Consider the strategy $\phi'_\Evei$ for \Eve defined by $\phi'_\Evei(\lambda)=\phi_\Evei((s,K)\cdot \action_\Evei\cdot\lambda)$, \emph{i.e.} $\phi_\Evei'$ mimics strategy $\phi_\Evei$. In particular, as $K'\notin \wknowledge{}$, it means that there is some $s''\in K'$ such that $\phi_\Evei'$ is not almost-surely winning for \Eve in $\game^\Ki$ from $(s'',K')$. 
Now, by considering the way the knowledge is updated, one concludes that there is some state $s'''\in\states$ such that $\ftrans^\Ki((s''',K),\action_\Evei,\action_\Adami))(s'',K')>0$ (as $s''$ belongs to a knowledge updated from $K$). Now, if one considers the strategy of \Adam that plays $\action_\Adami$ and then mimics a counter strategy against $\phi'_\Evei$, it follows that this strategy ensures \Adam to win with positive probability while playing against $\phi_\Evei$ starting from $(s''',K)$. This contradicts our initial assumption on $\phi_\Evei$ being almost-surely winning from all configuration of the form $(s,K)$ with $s\in K$. Hence $\allow{K}\neq \emptyset$.
\qed\end{proof}

We now consider a (well-defined) knowledge-based uniform memoryless strategy $\phi$ for \Eve on $K\in\wknowledge{}$ by letting 
$$\phi(K)(\action_\Evei)=
\begin{cases}
\dfrac{1}{|\allow{K}|} & \text{if }\action_\Evei\in \allow{K}\\
0 & \text{otherwise}
\end{cases}$$

The next proposition shows that $\phi$ is almost-surely winning for \Eve.


\begin{proposition}\label{prop:uniformstratiswinning}
The strategy $\phi$ is almost-surely winning for \Eve from states whose \Eve's knowledge is in $\wknowledge{}$.
\end{proposition}


\begin{proof}

In order to prove that $\phi$ is almost-surely winning one needs to show that it is almost-surely winning against \emph{any} observation-based strategy $\phi_\Adami$ of \Adam, \emph{i.e.} $\proba{(s_0,\{s_0\})}{\phi}{\phi_\Adami}(\objective)=1$. Note that once $\phi$ is fixed, and as it is a knowledge-based memoryless strategy, it induces a new game, denoted $\game_\phi$, where only \Adam makes choices, and where $\phi$'s choices are now part of the stochastic aspect of the game. More formally, $\game_\phi$ is the \onehalf-player \emph{saftey} game $\game_\phi=(\arena_\phi,\requiv_\Adami^\Ki,(s_0,\{s_0\}),\objective)$ where $\arena_\phi=\langle \states^\Ki,\actions_\Adami,\ftrans_\phi,\fstates^\Ki\rangle$ with
\begin{itemize}
\item 
$\ftrans_\phi((s,K),\action_\Adami)(s',K')= \displaystyle\sum_{\action_\Evei\in\actions_\Evei}\delta^\Ki((s,K),\action_\Evei,\action_\Adami)(s',K').\phi(K)(\action_\Evei)$
\end{itemize}

As a strategy for \Adam in $\game$ (equivalently in $\game^\Ki$) can be seen as well as a strategy in game $\game_\phi$ and \emph{vice versa}, while preserving the value of the game (against $\phi$ in $\game$), one derives the following fact.

\begin{fact}Strategy $\phi$ is almost-surely winning in $\game^\Ki$ if and only if the player has no positively winning observation-based strategy in $\game_\phi$.
\end{fact}

As $\game_\phi$ is a \onehalf-player safety concurrent game with imperfect information, one can use the previous results. In particular Lemma~\ref{lemma:safety1.5bis} implies that in order to prove Proposition~\ref{prop:uniformstratiswinning}, it suffices to prove that $\phi$ is winning against any strategy in $\game^\Ki$ that is obtained by mimicking an ultimately memoryless knowledge-based strategy of the player in $\game_\phi$\footnote{An ultimately memoryless knowledge-based strategy of the player in $\game_\phi$ is a finite memory strategy that uses as a memory a set of integers $\{1,\cdots,k\}$ (for the initial part) together with a set of knowledge. Note that a knowledge $\game_\phi$ is a subset of states, \emph{i.e.} is a subset of pairs made of a control state $s\in\states$ and of a subset $K\subseteq \states$ that represents a knowledge of \Eve in the previous game $\game$ (hence \Adam is not only computing the possible states he can be in but he is also computing the set of knowledges \Eve can have about the play. This somehow proves that to positively wins against $\phi$ he would use doubly exponential size memory). Nevertheless the only important thing is that it is a finite memory strategy, and that it can hence be translated in another finite memory strategy in $\game^\Ki$.}. 
As those strategies are finite memory strategies one derives the following fact.

\begin{fact}Strategy $\phi$ is almost-surely winning in $\game^\Ki$ if and only if it is almost-surely winning against any finite-memory observation-based strategy of \Adam.
\end{fact}

Fix such a finite memory strategy $\psi=(Move,Up,m_0)$ for \Adam (let $M$ be the finite memory used here). Now from $M$, $Up$ and $\arena^\Ki$ one can construct a new arena whose set of states is $S^\Ki\times M$ and where the $M$ component is updated accordingly to $Up$: the idea is just to make explicit the value of the memory at any stage and to update it explicitly in the arena. Next, if one modifies the equivalence relation of \Adam $\requiv_\Adami$ to only distinguish between those configurations that have a different memory content, and if one modifies $\requiv_\Evei$ so that \Eve has no information on the $M$ component, one obtains a new game in which any finite-memory strategy of \Adam with memory $M$ and update function $Up$ in the previous game is transformed into a memoryless observation-based strategy for him.

More formally, this leads to consider the arena \\ $\arena^\Ki_{Up}=\langle S^\Ki\times M,\actions_\Evei,\actions_\Adami,\delta^\Ki_{Up},\fstates^\Ki\times M \rangle$ where \\
$\delta^\Ki_{Up}((s,K,m),\sigma_\Evei,\sigma_\Adami))(s',K',m')=0$ if $m'\neq Up(m,[s']_{/_{\requiv_\Adami}},\action_\Adami)$ and\\
 $\delta^\Ki_{Up}((s,K,m),\sigma_\Evei,\sigma_\Adami))(s',K',m')=\delta^\Ki((s,K),\sigma_\Evei,\sigma_\Adami))(s',K')$ otherwise.

The new equivalence relations are given by $(s,K,m)\equiv_\Evei(s',K',m')$ iff $(s,K)\requiv^\Ki_\Evei(s',K')$ and $(s,K,m)\equiv_\Adami(s',K',m')$ iff $m=m'$. Let $\game^\Ki_{Up}$ be this new game. Any strategy for \Eve in $\game^\Ki$ can be seen as a strategy in $\game^\Ki_{Up}$ and vice versa.

Hence from now on we may only work in $\game^\Ki_{Up}$ and assume that $\psi$ is a memoryless knowledge-based strategy for \Adam and our goal is to prove that $\phi$ almost-surely wins against $\psi$ from any configuration in $\{(s,K,m_0) \mid K\in\wknowledge{} \text{ and }s\in K\}$. Actually, one will prove a slightly stronger result, namely that $\phi$ almost-surely wins against $\psi$ from any configuration in $\{(s,K,m) \mid m\in M,\ K\in\wknowledge{} \text{ and }s\in K\}$

We will first define an increasing sequence of subsets of almost-surely winning positions for \Eve in $\game^\Ki_{Up}$ and later we will prove that its limit is the set of all positions with a knowledge in $\wknowledge{}$ and that $\phi$ is actually an almost-surely winning strategy from those positions.
 For some configuration $(s,K,m)$ and some action $\action_\Evei\in\actions_\Evei$ and some distribution of actions $d$ in $\distribution{\actions_\Adami}$, we define $\post{\sigma_\Evei}{d}{(s,K,m)}$ as the set of all possible next states when \Eve plays $\action_\Evei$ and \Adam picks an action according to $d$ from $(s,K,m)$:
\begin{multline*}
\post{\sigma_\Evei}{d}{(s,K,m)}=\{(s',K',m')\mid \exists \action_\Adami \text{ s.t. } d(\action_\Adami)>0 \\ 
\text{ and } \delta^\Ki((s,K,m),\action_\Evei,\action_\Adami)((s',K',m'))>0\}
\end{multline*}
 Consider the following increasing sequence $(\rank{i})_{i\geq 0}$:
$\rank{0}=2^{F^\Ki}\times M$ consists of trivially winning positions and
 \begin{multline*}
\rank{i+1} = \rank{i} \cup
 \{(s,K,m)\mid K\in \wknowledge{} \\ 
 \text{ and }\exists\ \action_\Evei\in\allow{K}\text{ s.t. } \post{\action_{\Evei}}{Move(m)}{s,K,m}\cap \rank{i}\neq \emptyset\}
\end{multline*}

Let us denote by $\ranketoile$ the limit of the sequence $(\rank{i})_{i\geq 0}$. We claim that $\ranketoile=\{(s,K,m)\mid K\in\wknowledge{}\}$ and that $\phi$ is an almost-surely winning strategy for \Eve from those positions when \Adam plays accordingly to $\psi$.

The inclusion $\ranketoile\subseteq \{(s,K,m)\mid K\in\wknowledge{}\}$ is forced by the definition of $\ranketoile$. 
The fact that $\phi$ is an almost-surely winning strategy for \Eve from positions in $\ranketoile$ when \Adam plays accordingly to $\psi$ is a simple consequence of how $(\rank{i})_{i\geq 0}$ is defined and of Borel-Cantelli Lemma. Indeed from any configuration in $\rank{i}$, there is a non null probability to reach a final state in the next $i$ moves while playing $\phi$ against $\psi$ and moreover the play surely stays inside $\ranketoile$ while playing $\phi$ against $\psi$ (at least until some final state is visited): hence the probability, that a play starting from $\ranketoile$, in which \Eve follows $\phi$ and \Adam follows $\psi$, to never reach $F^\Ki\times M$ is null. 

In order to prove the other inclusion, we let $X=\{(s,K,m)\mid K\in\wknowledge{}\}\setminus\ranketoile$ and assume by contradiction that $X\neq\emptyset$. By definition, for any element $(s,K,m)\in X$ we have that $\forall \action_\Evei\in\allow{K}$, $\post{\action_{\Evei}}{Move(m)}{s,K,m}\subseteq X$. This means in particular that following $\phi$ from such a configuration, and if \Adam plays accordingly to $\psi$, then \Eve surely looses as the play surely stays in $X$ and $X\cap 2^{F^\Ki}\times M=\emptyset$. Now we claim that the same holds if one replaces $\phi$ by any almost-surely winning strategy for \Eve, leading to a contradiction. Indeed consider an almost-surely winning strategy $\phi^\AS$ for \Eve. Then we claim the following fact.

\begin{fact}\label{property:stayinallow} Let $\lambda$ be a partial play consisting only of configurations in $X$. Assume that, for some strategy $\psi'$ of \Adam, $\lambda$ is a possible partial play accordingly to both $\phi^\AS$ and $\psi'$ (more formally, $\proba{(s,K,m)}{\phi}{\psi'}(\cone{\lambda})>0$ where $(s,K,m)$ denotes the initial configuration of $\lambda$). Then for any action $\action_\Evei\in\actions_\Evei$, one has $\phi^\AS(\lambda)(\action_\Evei)>0$ only if $\action_\Evei\in\allow{K'}$ where $K'$ denotes \Eve's knowledge in the last configuration of $\lambda$.
\end{fact}

\begin{proof}
The proof is by contradiction. Consider some $\lambda,\phi^\AS,\psi,\action_\Evei$ violating the property. Then one can find an equivalent play $\lambda'\requiv_\Evei \lambda$ such that $\lambda'$ ends in a configuration $(s',K',m')$ and there is an action $\action_\Adami\in\actions_\Adami$ such that\\ $\ftrans^\Ki((s',K',m'),\action_\Evei,\action_\Adami)((s'',K'',m''))>0$ for some $K''\notin \wknowledge{}$ (existence of $\lambda'$ follows by the construction of the knowledge arena). Now consider the strategy of \Adam that first mimics $\psi$ and then after $|\lambda'|$ moves plays $\action_\Adami$ and then plays accordingly to a strategy ensuring from $(s'',K'',m'')$ that \Eve's does not surely wins (such a strategy exists as $K''\notin \wknowledge{}$): then against this strategy $\phi^\AS$ is not almost-surely winning, leading a contradiction
\finpreuve{Fact \ref{property:stayinallow}}
\end{proof}

Now one is ready to conclude. Assume \Adam plays accordingly to $\psi$ and \Eve plays accordingly to some almost-surely winning strategy $\phi^\AS$. Then it follows from Fact \ref{property:stayinallow} and definition of $X$ that a play starting in $X$ stays forever in $X$, hence never visits $F^\Ki\times M$ and contradicting the hypothesis that $\phi^\AS$ is almost-surely winning. Therefore $X=\emptyset$, which concludes the proof of 
Proposition~\ref{prop:uniformstratiswinning}.\qed
\end{proof}

Now one can prove Theorem~\ref{theo:reachability}. First \Eve almost-surely wins in $\game$ if and only if she almost-surely wins in $\game^\Ki$ if and only if $\{s_0\}\in \wknowledge{}$, \emph{i.e.} (using Proposition \ref{prop:uniformstratiswinning}) if and only if \Eve has a knowledge-only uniform strategy in $\game^\Ki$. Now, to decide whether \Eve almost-surely wins in $\game$, it suffices to check, for any possible knowledge-only uniform strategy $\phi$ for her, whether it is almost-surely winning. Once $\phi$ is fixed, it leads, from Adam's point of view, to a \onehalf-player safety game $\game_\phi$ where the player positively wins if and only if $\phi$ is not almost-surely winning.
Hence Lemma \ref{lemma:safety1.5} implies that deciding whether $\phi$ is almost-surely winning can be done in time exponential in the size of $\game_\phi$, which itself is of exponential size in $|\states|$. Hence deciding whether a knowledge-only uniform strategy for \Eve is winning can be done in doubly exponential time (in the size of $|\states|$). The set of knowledge-only uniform strategies for \Eve is finite and its size is doubly exponential in the size of the game. Hence the overall procedure, that tests every possible such strategies, requires doubly exponential time. As effectivity is immediate, this concludes the proof of Theorem~\ref{theo:reachability}.

\subsection{Complexity Lower Bound}

The naive underlying algorithm of Theorem~\ref{theo:reachability} turns out to be optimal.

\begin{theorem}\label{theo:lowerbound}
Deciding whether \Eve almost-surely wins a concurrent game with imperfect information is a $2$-\textsc{ExpTime}-complete problem.
\end{theorem}

\begin{proof}[sketch]
The proof is a generalisation of a similar result given in \cite{CDHR07} showing \textsc{ExpTime}-hardness of concurrent games \emph{only one} player is imperfectly informed. The idea is to simulate an alternating exponential space Turing machine (without input). We design a game where the players describe the run of such a machine: transitions from existential (\emph{resp.} universal) states are chosen by \Eve (\emph{resp.} Adam) and \Adam is also in charge of describing the successive configurations of the machine. To prevent him from cheating, \Eve can secretly mark a cell of the tape, and latter check whether it was correctly updated (if not she wins). As she cannot store the exact index of the cell (it is of exponential size), she could cheat in the previous phase: hence \Adam secretly marks some bit and one recall the value of the corresponding bit of the index of the marked cell: this bit is checked when \Eve claims that \Adam cheated (if it is wrong then she is loosing). \Eve also wins if the described run is accepting. \Eve can also restart the computation whenever she wants (this is useful when she cannot prove that \Adam cheated): hence if the machine accepts the only option for \Adam is to cheat, and \Eve will eventually catch him with probability one. Now if the machine does not accept, the only option for \Eve is to cheat, but it will be detected with positive probability.\qed
\end{proof}

\section{B\"uchi Objectives}

We now consider the problem of deciding whether \Eve almost-surely wins a B\"uchi game and establish the following result. 

\begin{theorem}\label{theo:buchi}
For any B\"uchi concurrent game with imperfect information, one can decide, in doubly exponential time, whether \Eve has an almost-surely winning strategy. If \Eve has such a strategy then she has a knowledge-based uniform memoryless strategy, and such a strategy can be effectively constructed.
The doubly exponential time complexity bound is optimal.
\end{theorem}

The results and techniques are similar to the one for reachability games.
In particular, we need first to establish an intermediate result about positive winning in \onehalf-player co-B\"uchi game with imperfect information.

\subsection{Positively winning in \onehalf-player co-B\"uchi game with imperfect information}

We have the following lemma whose statement and proof are very similar to the one of Lemma~\ref{lemma:safety1.5} except that now the winning states are those connected by \emph{any kind} of path to a surely winning state).

\begin{lemma}\label{lemma:coBuchi1.5}
Consider an \onehalf-player co-B\"uchi game with imperfect information. Assume that the player has an observation-based strategy that is positively winning.
Then she also has an observation-based finite memory strategy that is positively winning.
Moreover, both the strategy and the set of positively winning states can be computed in time $\mathcal{O}(2^{|\states|})$.
\end{lemma}

Lemma \ref{lemma:coBuchi1.5} is a direct consequence of the following slightly more precise lemma.

\begin{lemma}\label{lemma:coBuchi1.5bis}
Consider an \onehalf-player co-B\"uchi game with imperfect information. Assume that the player has an observation-based strategy that is positively winning.
Then she also has an ultimately knowledge-based memoryless strategy that is positively winning.
Moreover, both the strategy and the set of positively winning states can be computed in time $\mathcal{O}(2^{|\states|})$.\rm
\end{lemma}

\begin{proof}
Fix an arena $\arena=\langle \states,\actions,\ftrans\rangle$ together with a co-B\"uchi game $\game=(\arena,\requiv,s_0,F)$ (recall here that a play is winning if and only if it visits finitely often the set $F$).

In the sequel we will be interested in computing the set $\poswstates$ of those states $s\in S$ from which the player has an observation-based strategy that  is positively winning.

The proof is very similar to the one of Lemma~\ref{lemma:safety1.5bis}, and therefore we reuse the notations as well as intermediate results. Recall that to win a safety game with positive probability the player needs to reach a state from which she is sure not to visit a final state (those states were denoted $\wstates$). In order to reach such a state, the player plays randomly on some fixed number of initial moves. Then she bets that she reached some good state $s\in\wstates$ and plays as if it is the case and in such a way to surely win (actually the state $s$ is guessed from the very beginning and is part of the strategy). As this second stage can be done mimicking a knowledge-based memoryless strategy, this leads to an ultimately knowledge-based memoryless  strategy. 

An important point is that the state $s$ should be reached without visiting a final state in the meantime. If the game is equipped now with a co-B\"uchi objective, it no longer matter. This remark leads to the following definition.
Consider the following (increasing and bounded) sequence:
$$\begin{cases}
\hat{W}_0 = \wstates\\
\hat{W}_{i+1} = \hat{W}_{i} \cup \{s\in\states\mid \exists \action\in\actions \text{ and }t\in \hat{W}_i \text{ s.t. }\delta(s,\action)(t)>0\}\\
\end{cases}
$$
Let $\hat{W}$ be the limit of the sequence $(\hat{W}_{i})_{i\geq 0}$: it consists exactly of those states from which the player has a strategy ensuring her to reach, in at most $|\states|$ moves, a state in $\wstates$ with some positive probability. Also note that the corresponding strategy is the one that plays randomly (with equiprobability) any action in $\actions$. Hence the only difference between the set $\hat{W}$ and the set $W$ defined in the proof of Lemma~\ref{lemma:safety1.5bis} is that one allows to have final states in the fixpoint definition of $\hat{W}$.

The states in $\hat{W}$ are actually those from which the player can wins the co-B\"uchi game with a positive probability.

\begin{fact}\label{fact:poswinningstates1-2coBuchi}
The following equality holds: $\hat{W}=\poswstates$.
\end{fact}

\begin{proof}
The inclusion $\hat{W}\subseteq \poswstates$ is rather immediate: from some state $s\in \hat{W}$, the player should first play randomly on the first $i$ rounds (where $i$ is the smallest integer such that $t\in \hat{W}_i$ in the previously defined sequence) and then play as if it was in some state $t\in\wstates$ (where $t$ is some reachable state from $s$ in $i$ moves accordingly to the definition of the sequence $(\hat{W}_i)_{i\geq 0})$. This last step is done using an observation-based strategy that mimics the one coming with the construction of $\wknowledge{}$: the only trick here is that the player "reset" the knowledge to be $\{t\}$.

Consider now the converse inclusion: $\poswstates \subseteq \hat{W}$. Let $s$ be some state in $\poswstates$ and assume, by contradiction, that $s\notin \hat{W}$. In particular a play starting from $s$ will never go through a state in $\wstates$. We claim that playing accordingly to some knowledge-based strategy is almost-surely loosing for the player: indeed, as a consequence of the definition of $\wstates$ (and of $\wknowledge{}$), from a state $t\notin\wstates$, playing a knowledge-based strategy, the probability of visiting a final state in the next $2^{|\states|}$ moves is some $\epsilon>0$. Moreover such a play stays outside of $\wstates$ forever. Hence, for any $k>0$, using Borel-Cantelli Lemma, the probability of such a play to visit at most $k$ final states is $0$ hence implying that the probability of going finitely often through $F$ is $0$ too. This means that $s$ is surely loosing for the player and contradicts the initial hypothesis of $s\in \poswstates$. Hence, $s\in \hat{W}$ which concludes the proof.
\finpreuve{Fact~\ref{fact:poswinningstates1-2coBuchi}}
\end{proof}

Now one can conclude the proof of Lemma~\ref{lemma:coBuchi1.5}: the existence of the ultimately knowledge-based memoryless strategy follows from Fact~\ref{fact:poswinningstates1-2}. Effectivity as well as complexity comes from the constructive way of defining $\hat{W}$ (and other intermediate objects) by means of fixpoint computations.
\finpreuve{Lemma~\ref{lemma:coBuchi1.5}}
\end{proof}

\subsection{Proof of Theorem \ref{theo:buchi}}

We now turn to the proof of Theorem \ref{theo:buchi}. Again, the key idea is to prove that the strategy that plays randomly inside the almost-surely winning region is an almost-surely winning strategy. 

For the rest of this section, we fix a concurrent game with imperfect information $\game= (\arena,\requiv_\Evei,\requiv_\Adami,s_0,\objective)$ equipped with a B\"uchi objective $\objective$. We set $\arena = \tuple{\states,\actions_\Evei,\actions_\Adami,\ftrans,\fstates}$. We also consider $\game^\Ki=(\arena^\Ki,\requiv_{\Evei}^\Ki,\requiv_{\Adami}^\Ki,(s_0,\{s_0\}),\objective^\Ki)$ to be the corresponding knowledge game.

The proof follows the same line as the one to prove Theorem~\ref{theo:reachability}. The main idea is again to prove that the strategy that stay inside the configuration with an almost-surely winning knowledge is almost-surely winning.

Again, we let
\begin{multline*}
	\wknowledge{}=\{K\in 2^\states\mid \exists \varphi_\Evei \text{ knowledge-based strategy for \Eve s.t. } \varphi_\Evei \text{ is}  \\ 
	\text{almost-surely winning for \Eve in } \game^\Ki \text{ from any }(s,K) \text{ with } s\in K\}
\end{multline*}
be the set of equivalent classes (with respect to $\requiv_{\Evei}^\Ki$) made only by almost-surely winning states for \Eve (note here that we require that the almost-surely winning strategy is the same for all configurations with the same knowledge). 
For every knowledge $K\in\wknowledge{}$ we define 
\begin{multline*}
	\allow{K}=\{\action_\Evei\in\actions_\Evei
	\mid \forall s\in K, \forall \action_\Adami\in\actions_\Adami,\\
	 \ftrans^\Ki((s,K),\action_\Evei,\action_\Adami)((s',K))>0 \Rightarrow K'\in\wknowledge{}\}
\end{multline*}

We then get a result similar to Proposition \ref{proposition:allownonempty} (the proof is exactly the same as the one of Proposition \ref{proposition:allownonempty}).

\begin{proposition}\label{proposition:allownonemptyB}
For every knowledge $K_\Evei\in\wknowledge{\Evei}$, $\allow{K_{\Evei}}\neq \emptyset$.
\end{proposition}

Again, we define now a knowledge-based uniform memoryless strategy $\phi$ for \Eve on $K\in\wknowledge{}$ by letting 
$$\phi(K)(\action_\Evei)=
\begin{cases}
\dfrac{1}{|\allow{K}|} & \text{if }\action_\Evei\in \allow{K}\\
0 & \text{otherwise}
\end{cases}$$

The next proposition shows that $\phi$ is almost-surely winning for \Eve.

\begin{proposition}\label{prop:uniformstratiswinningB}
The strategy $\phi$ is almost-surely winning for \Eve from states whose \Eve's knowledge is in $\wknowledge{\Evei}$.
\end{proposition}

\begin{proof}

In order to prove that $\phi$ is almost-surely winning one needs to show that it is almost-surely winning against \emph{any} observation-based strategy $\phi_\Adami$ of \Adam, \emph{i.e.} $\proba{(s_0,\{s_0\})}{\phi}{\phi_\Adami}(\objective)=1$. Again, as in the reachability case, once $\phi$ is fixed, and as it is a knowledge-based memoryless strategy, it induces a \onehalf-player co-B\"uchi game denoted $\game_\phi$ and defined exactly as in the proof of Proposition~\ref{prop:uniformstratiswinning}.

As a strategy for \Adam in $\game$ (equivalently in $\game^\Ki$) can be seen as well as a strategy in game $\game_\phi$ and \emph{vice versa}, while preserving the value of the game (against $\phi$ in $\game$), one derives the following fact.

\begin{fact}Strategy $\phi$ is almost-surely winning in $\game^\Ki$ if and only if the player has no positively winning observation-based strategy in the co-Buchi $\game_\phi$.
\end{fact}

As $\game_\phi$ is a \onehalf-player co-B\"uchi concurrent game with imperfect information, one can use Lemma~\ref{lemma:coBuchi1.5bis} to conclude that, in order to prove Proposition~\ref{prop:uniformstratiswinning}, it suffices to prove that $\phi$ is winning against any \emph{finite-memory} observation-based strategy of \Adam in $\game^\Ki$.

\begin{fact}Strategy $\phi$ is almost-surely winning in $\game^\Ki$ if and only if it is almost-surely winning against any finite-memory observation-based strategy of \Adam.
\end{fact}

Fix such a finite memory strategy $\psi=(Move,Up,m_0)$ for \Adam (let $M$ be the finite memory used here). As in the proof of Proposition~\ref{prop:uniformstratiswinning} it leads to define a new game $\game^\Ki_{Up}$ (we keep here the same notations / definitions). Again, we may assume that $\psi$ is a memoryless knowledge-based strategy for \Adam and our goal is to prove that $\phi$ almost-surely wins against $\psi$ from any configuration in $\{(s,K,m_0) \mid K\in\wknowledge{} \text{ and }s\in K\}$. Actually, one will prove a slightly stronger result, namely that $\phi$ almost-surely wins against $\psi$ from any configuration in $\{(s,K,m) \mid m\in M,\ K\in\wknowledge{} \text{ and }s\in K\}$.

We will first define an increasing sequence of subsets of almost-surely winning positions for \Eve in $\game^\Ki_{Up}$ and later we will prove that its limit is the set of all positions with a knowledge in $\wknowledge{}$ and that $\phi$ is actually an almost-surely winning strategy from those positions.
 For some configuration $(s,K,m)$ and some action $\action_\Evei\in\actions_\Evei$ and some distribution of actions $d$ in $\distribution{\actions_\Adami}$, we define $\post{\sigma_\Evei}{d}{(s,K,m)}$ as the set of all possible next states when \Eve plays $\action_\Evei$ and \Adam picks an action according to $d$ from $(s,K,m)$:
\begin{multline*}
\post{\sigma_\Evei}{d}{(s,K,m)}=\{(s',K',m')\mid \exists \action_\Adami \text{ s.t. } d(\action_\Adami)>0 \\ 
\text{ and } \delta^\Ki((s,K,m),\action_\Evei,\action_\Adami)((s',K',m'))>0\}
\end{multline*}
 Consider the following increasing sequence $(\rank{i})_{i\geq 0}$:
$\rank{0}=(2^{F^\Ki}\cap\wknowledge{})\times M$ consists of trivially winning positions and
 \begin{multline*}
\rank{i+1} = \rank{i} \cup
 \{(s,K,m)\mid K\in \wknowledge{} \\ 
 \text{ and }\exists\ \action_\Evei\in\allow{K}\text{ s.t. } \post{\action_{\Evei}}{Move(m)}{s,K,m}\cap \rank{i}\neq \emptyset\}
\end{multline*}

Let us denote by $\ranketoile$ the limit of the sequence $(\rank{i})_{i\geq 0}$. We claim that $\ranketoile=\{(s,K,m)\mid K\in\wknowledge{}\}$ and that $\phi$ is an almost-surely winning strategy for \Eve from those positions when \Adam plays accordingly to $\psi$.

The inclusion $\ranketoile\subseteq \{(s,K,m)\mid K\in\wknowledge{}\}$ is forced by the definition of $\ranketoile$. 
The fact that $\phi$ is an almost-surely winning strategy for \Eve from positions in $\ranketoile$ when \Adam plays accordingly to $\psi$ is a simple consequence of how $(\rank{i})_{i\geq 0}$ is defined and of Borel-Cantelli Lemma. 
Indeed from any configuration in $\rank{i}$, there is a non null probability to reach a final state in the next $i$ moves while playing $\phi$ against $\psi$ and moreover the play surely stay inside $\ranketoile$ while playing $\phi$ against $\psi$: hence for any $k\geq 0$, the probability, that a play starting from $\ranketoile$, in which \Eve follows $\phi$ and \Adam follows $\psi$, visits at most $k$ time a final configuration is null. Therefore the probability of going infinitely often through a final state is $1$, meaning that $\phi$ is almost-surely winning in $\ranketoile$.

In order to prove the other inclusion, we let $X=\{(s,K,m)\mid K\in\wknowledge{}\}\setminus\ranketoile$ and assume by contradiction that $X\neq\emptyset$. By definition, for any element $(s,K,m)\in X$ we have that $\forall \action_\Evei\in\allow{K}$, $\post{\action_{\Evei}}{Move(m)}{s,K,m}\subseteq X$. This means in particular that following $\phi$ from such a configuration, and if \Adam plays accordingly to $\psi$, then \Eve surely looses as the play surely stay in $X$ and $X\cap 2^{F^\Ki}\times M=\emptyset$. Now we claim that the same holds if one replaces $\phi$ by any almost-surely winning strategy for \Eve, leading to a contradiction. Indeed consider an almost-surely winning strategy $\phi^\AS$ for \Eve. Then we have the following fact (whose proof is omitted as it is exactly the same as the one of Fact~\ref{property:stayinallow}).

\begin{fact}\label{property:stayinallowB} Let $\lambda$ be a partial play consisting only of configurations in $X$. Assume that, for some strategy $\psi'$ of \Adam, $\lambda$ is a possible partial play accordingly to both $\phi^\AS$ and $\psi'$ (more formally, $\proba{(s,K,m)}{\phi}{\psi'}(\cone{\lambda})>0$ where $(s,K,m)$ denotes the initial configuration of $\lambda$). Then for any action $\action_\Evei\in\actions_\Evei$, one has $\phi^\AS(\lambda)(\action_\Evei)>0$ if and only if $\action_\Evei\in\allow{K'}$ where $K'$ denotes \Eve's knowledge in the last configuration of $\lambda$.
\end{fact}

Now one is ready to conclude. Assume \Adam plays accordingly to $\psi$ and \Eve plays accordingly to some almost-surely winning strategy $\phi^\AS$. Then it follows from Fact \ref{property:stayinallowB} and definition of $X$ that a play starting in $X$ stays forever in $X$, hence never visits $F^\Ki\times M$ and contradicting the hypothesis that $\phi^\AS$ is almost-surely winning. Therefore $X=\emptyset$, which concludes the proof of 
Proposition~\ref{prop:uniformstratiswinningB}\qed
\end{proof}

\bigskip

Now one concludes the proof of Theorem~\ref{theo:buchi} exactly as for the proof of Theorem~\ref{theo:reachability}. The 2-ExpTime hardness lower bound follows from the fact that it already holds for reachbility objective (Theorem~\ref{theo:lowerbound}).

\section{Discussion}\label{section:discussion}

The main contribution of this paper is to prove that one can decide whether \Eve has an almost-surely winning strategy in a concurrent game with imperfect information equipped with a reachability objective or a B\"uchi objective. 

A natural question is whether this result holds for other objectives, in particular for co-B\"uchi objectives. In a recent work \cite{BBG08}, Baier \emph{et al.} 
established undecidability of the emptiness problem for probabilistic B\"uchi automata on infinite words. Such an automaton can be simulated by a \onehalf-player imperfect information game: the states of the game are the one of the automaton, they are all equivalent for the player, and therefore an observation based strategy is an infinite word. Hence a pure (\emph{i.e.} non-randomised) strategy in such a game coincide with an input word for the automaton. From this fact, Baier \emph{et al.} derived that it is undecidable whether, in a \onehalf-player co-B\"uchi game with imperfect information, \Eve has an almost-surely winning \emph{pure} strategy. 

One can also consider the stochastic-free version of this problem (an arena is \defin{deterministic} iff $\delta(q,\action_\Evei,\action_\Adami)(q')\in\{0,1\}$ for all $q,q',\action_\Evei,\action_\Adami$) and investigate whether one can decide if \Eve has an almost-surely winning strategy in a \emph{deterministic} game equipped with a co-B\"uchi objective. We believe that the \onehalf-player setting can be reduced to this new one, hence allowing to transfer undecidability results \cite{Flo}. An even weaker model to consider is the stochastic-free model in which \Adam has perfect information about the play \cite{CDHR07}. 

It may happen that \Eve has no almost-surely winning strategy while having a family $(\phi_\epsilon)_{0<\epsilon<1}$ of strategies such that $\phi_\epsilon$ ensures to win with probability at least $1-\epsilon$. Such a family is called \defin{limit-surely winning}. Deciding existence of such families is a very challenging problem: indeed, in many practical situations, it is satisfying enough if one can control the risk of failing. Even if those questions have been solved for perfect information games \cite{dAH00}, as far as we know, there has not been yet any result obtained in the imperfect information setting.

Even if the algorithms provided in this paper are "optimal", they are rather naive (checking all strategies for \Eve may cost a lot in practice). Hence, one should look for fixpoint-based algorithms as the one studied in \cite{CDHR07}: it would be of great help for a symbolic implementation, and it could also be a useful step toward a solution of the problem of finding limit-surely winning strategies. Note that there are already efficient techniques and tools for finding \emph{sure} winning strategies in subclasses of concurrent games with imperfect information \cite{BCDHR08,BCDDH09}.

\medskip
\noindent\textbf{Acknowledgements. }The authors want to warmly thank Krishnendu Chatterjee and Laurent Doyen for pointing out an important mistake in a previous version of this work originally published in ICALP'09 \cite{GS09}. Indeed, in this previous work we additionally assumed that the players did not observe the actions they played and we proposed a variation of the constructions for this richer setting. It turned out that the constructions were wrong. We refer to \cite{CD11} for more insight on this point as well as for an answer to this more general problem.


\begin{thebibliography}{10}


\bibitem{dAH00}
L.~de~Alfaro and T.A. Henzinger.
\newblock Concurrent omega-regular games.
\newblock In {\em Proceedings of LICS'00}, pages 141--154, 2000.

\bibitem{HdAK97}
L.~de~Alfaro, T.A. Henzinger, and O.~Kupferman.
\newblock Concurrent reachability games.
\newblock {\em Theoretical Computer Science}, 386(3):188--217, 2007.

\bibitem{BBG08}
C.~Baier, N.~Bertrand, and M.~Gr{\"o}{\ss}er.
\newblock On decision problems for probabilistic b{\"u}chi automata.
\newblock In {\em Proceedings of FoSSaCS 2008}, volume 4962 of {\em LNCS},
  pages 287--301. Springer, 2008.

\bibitem{PrivateGimbert}
N.~Bertrand, B.~Genest, and H.~Gimbert.
\newblock Qualitative Determinacy and Decidability of Stochastic Games with Signals.
\newblock In {\em Proceedings of LICS 2009}, pages 319--328. IEEE, 2009.


\bibitem{BCDDH09}
D.~Berwanger, K.~Chatterjee, M.~De~Wulf, L.~Doyen, and T.A. Henzinger.
\newblock Alpaga: A tool for solving parity games with imperfect information.
\newblock In {\em Proceedings of TACAS 2009}, volume 5505 of {\em LNCS},
  pages 58-61. Springer, 2009.

\bibitem{BCDHR08}
D.~Berwanger, K.~Chatterjee, L.~Doyen, T.A. Henzinger, and S.~Raje.
\newblock Strategy construction for parity games with imperfect information.
\newblock In {\em Proceedings of CONCUR 2008}, volume 5201 of {\em LNCS}, pages
  325--339. Springer, 2008.

\bibitem{chatterjeePHD}
K.~Chatterjee.
\newblock {\em {Stochastic $\omega$-Regular Games}}.
\newblock PhD thesis, University of California, 2007.

\bibitem{CDHR07}
K.~Chatterjee, L.~Doyen, T.A. Henzinger, and J.-F. Raskin.
\newblock Algorithms for omega-regular games with imperfect information.
\newblock {\em Logical Methods in Computer Science}, 3(3), 2007.

\bibitem{CD11}
K.~Chatterjee, L.~Doyen.
\newblock  Partial-Observation Stochastic Games: How to Win when Belief Fails.
\newblock Preprint, http://arxiv.org/abs/1107.2141, July 2011.


\bibitem{Flo}
F.~Horn.
\newblock Private communication.
\newblock February 2009.

\bibitem{CH05}
K.~Chatterjee and T.A. Henzinger.
\newblock Semiperfect-information games.
\newblock In {\em Proceedings of FST\&TCS 2005}, volume 3821 of {\em LNCS},
  pages 1--18. Springer, 2005.

\bibitem{dag01}
E.~Gr{\"a}del, W.~Thomas, and Th. Wilke, editors.
\newblock {\em Automata, Logics, and Infinite Games: A Guide to Current
  Research}, volume 2500 of {\em LNCS}. Springer, 2002.

\bibitem{GH82}
Y.~Gurevich and L.~Harrington.
\newblock Trees, automata, and games.
\newblock In {\em Proceedings of STOC 1982}, pages 60--65, 1982.

\bibitem{GS09}
V.~Gripon and O. Serre.
\newblock Qualitative Concurrent Games with Imperfect Information
\newblock In {\em Proceedings of ICALP 2009}, volume 5556 of {\em LNCS},
  pages 200--211. Springer, 2009.


\bibitem{Flo}
F.~Horn.
\newblock Private communication.
\newblock February 2009.

\bibitem{Paz71}
A.~Paz.
\newblock {\em {Introduction to probabilistic automata}}.
\newblock Academic Press New York, 1971.

\bibitem{RW87}
P.J. Ramadge and W.M. Wonham.
\newblock {Supervisory Control of a Class of Discrete Event Processes}.
\newblock {\em SIAM Journal on Control and Optimization}, 25:206, 1987.

\bibitem{Reif84}
J.H. Reif.
\newblock The complexity of two-player games of incomplete information.
\newblock {\em Journal of Computer and System Sciences}, 29(2):274--301, 1984.

\end{thebibliography}
\end{document}